%% file: GovRotFCP.tex
\documentclass[a4paper,10pt]{article}

\usepackage[final,obeyFinal]{todonotes}

\input{macros.tex}

\begin{document}
\title{Is Free Choice Permission Admissible in Classical Deontic Logic?}
\author{Guido Governatori\\ 
Data61, CSIRO, Australia
\and Antonino Rotolo\\
University of Bologna, Italy}
\date{}

\maketitle

\begin{abstract}
In this paper, we explore how, and if, free choice permission (\ref{eq:FCP}) can be accepted when we consider deontic conflicts between certain types of permissions and obligations. As is well known, \ref{eq:FCP} can license, under some minimal conditions, the derivation of an indefinite number of permissions. We discuss this and other  drawbacks and present six Hilbert-style classical deontic systems admitting a guarded version of \ref{eq:FCP}. The systems that we present are not too weak from the inferential viewpoint, as far as permission is concerned, and do not commit to weakening any specific logic for obligations.
\end{abstract}

%
\section{Introduction and Background}\label{sec:Introduction}

A significant part of the literature in deontic logic revolves around the discussions of puzzles and paradoxes which show that certain logical systems are not acceptable---typically, this happens with deontic $\Lg{KD}$, i.e., Standard Deontic Logic ($\Lg{SDL}$)---or which suggest that obligations and permissions should enjoy some desirable properties.

One well-known puzzle is the the so-called Free Choice Permission paradox, which was originated by the following remark by von Wright in \cite[p. 21]{vonWright1968-VONAEI-3}:

\begin{quote}
``On an ordinary understanding of the phrase `it is permitted that', the formula  `$\PERM (p \vee q)$' seems to entail `$\PERM p \wedge \PERM q$'. If I say to somebody `you may work or relax' I normally mean that the person addressed has my permission to work and also my permission to relax. It is up to him to choose between the two alternatives.''
\end{quote}

Usually, this intuition is formalised by the following schema:

\begin{equation}\tag{\bf{FCP}}
	\label{eq:FCP}
\PERM (p \vee q) \to (\PERM p \wedge \PERM q)
\end{equation}

Many problems have been discussed in the literature around \ref{eq:FCP}: for a  comprehensive overview, discussion, and some solutions, see \cite{Hansson:PERM,Dong:2017,sep-logic-deontic}.

Three basic difficulties can be identified, among the others \cite[][p. 43]{Dong:2017}:
\begin{itemize}
	\item \textbf{Problem 1: Permission Explosion Problem} -- ``That if anything is permissible, then everything is, and thus it would also be a theorem that nothing is obligatory,'' \cite{sep-logic-deontic}, for example ``If you may order a soup, then it is not true that you ought to pay the bill'' \cite{Asher2005};
	\item \textbf{Problem 2: Closure under Logical Equivalence Problem} -- ``In its classical form \textbf{FCP} entails that classically equivalent formulas can be substituted to the scope of a permission operator. This is also implausible: It is permitted to eat an apple or not iff it is permitted to sell a house or not'';
	\item \textbf{Problem 3: Resource Sensitivity Problem} -- 	``Many deontic logics become resource-insensitive in the presence of \textbf{FCP}. They validate inferences of the form `if the patient with stomach trouble is allowed to eat one cookie then he is allowed to eat more than one'\,''.
\end{itemize}

In this paper, we focus on another basic problem: how, and if, \textbf{FCP} can be accepted when we have incompatibilities between certain varieties of permissions and prohibitions/obligations. The issue is that since Problem 1 licenses the derivation that anything is permitted provided that something is permitted, no prohibition/obligation is allowed, otherwise we get an inconsistency \cite{sep-logic-deontic}. In doing so, we will offer simple logics that take two of the three  problems above into account. 

The layout of the paper is as follows. The remainder of this section briefly comments on the three major problems mentioned above: the Permission Explosion Problem (Section \ref{sec:explosion}), the Closure under Logical Equivalence Problem (Section \ref{sec:RE}), and the Resource Sensitivity Problem (\ref{sec:sensitivity}). Section \ref{sec:purpose} illustrates the theoretical intuitions and assumptions that we adopt to analyse free choice permission. In particular, we assume the distinction between norms and obligations/permissions, and we study the role of deontic incompatibilities, the duality principle, and why free choice permission is strong permission. Section \ref{sec:related_work} reviews in some detail two related works that have direct implications for our proposal. Finally, Section \ref{sec:system} presents some minimal deontic systems, six Hilbert-style deontic systems admitting guarded variants of \ref{eq:FCP}: the systems that we present are not too weak from the inferential viewpoint, as far as permission is concerned, and do not commit to weakening any specific logic for obligations. Some conclusions end the paper. An appendix offers proofs of the formal properties of the proposed systems presented in Section \ref{sec:system}.

\subsection{Problem 1: Permission Explosion Problem}\label{sec:explosion}

One of the most acute problems springing from \ref{eq:FCP} is obtained in $\SDL$, where, if at least one obligation $\OBL p$ is true, then by necessitation and propositional logic, we get $\OBL (p \vee q)$. Since axiom $\ax{D}$ is in $\SDL$, i.e $\OBL p \to \neg \OBL \neg p$ is valid, we trivially obtain $\neg \OBL \neg (p \vee q)$, thus, assuming the Duality principle
\begin{equation}\tag{\text{\bf{Duality}}}
	\label{eq:duality}
\PERM =_{\mathit{def}} \neg \OBL \neg
\end{equation}
we derive through \ref{eq:FCP} that  $\PERM q$. Hence, $\SDL$ licenses that, if something is obligatory, then everything is permitted.

However, a careful analysis shows that this undesired result is not strictly due to $\SDL$ as such, but to adopting
any monotonic modal deontic logic \cite{Chellas1980}, i.e. any system just equipped with inference rule $\ax{RM}$:
\begin{equation}\tag{\bf{RM}}
	\label{eq:RM}
  \frac{\vdash p\to q}{\vdash \OBL p\to \OBL q}.
\end{equation}
or, alternatively with
\begin{equation}\tag{\bf{RE}}
	\label{eq:RE}
  \frac{\vdash p\equiv q}{\vdash \OBL p\equiv \OBL q}.
\end{equation}
plus the following axiom schema
\begin{equation}\tag{\bf{M}}
	\label{eq:M}
\OBL (a \wedge b) \to (\OBL a \wedge \OBL b).
\end{equation}

Indeed, assume Classical Propositional Logic ($\Lg{CPL}$), \textbf{FCP}, and \textbf{RM} for $\PERM$\footnote{I.e., 
\begin{equation}\tag{\bf{RM-P}}
	\label{eq:RM-P}
  \frac{\vdash p\to q}{\vdash \PERM p\to \PERM q}.
\end{equation}
Indeed, it is standard result that every system closed under \ref{eq:RM} for an operator is closed under the rule of the dual of the operator \cite[cf.][p. 238--239, 243]{Chellas1980}. We will use $\ax{RM}$ to refer in general to the rule $\vdash p\to q/ \vdash \Box p\to \Box q$ for any modal operator $\Box$.} and consider the following derivation:
\[
\begin{array}{lll}
1. & p \to (p\vee q) & \Lg{CPL}\\
2. & \PERM p \to \PERM (p\vee q) & 1, \ax{RM}\\
3. & \PERM p \to (\PERM p\wedge \PERM q) & 2, \ax{FCP}, \Lg{CPL}\\
4. & \PERM p\to \PERM q & 3, \Lg{CPL}\\
\end{array}
\]
In this context, it is enough if we have that $\PERM p$ is true to derive that any other permission $\PERM q$, i.e., $\PERM p \vdash \PERM q$ for any $p,q$.
%
Whenever $\ax{FCP}$ is accepted, such a problem strictly depends on the characteristic schemas and inference rules of monotonic modal logics, as the above derivation---or a simple semantic analysis---shows. Hence, permission explosion is not a problem of $\Lg{SDL}$, but of any weaker modal deontic logic which is at least closed under classical implication or which is closed under logical equivalence and allows for the distribution of $\PERM$ over implication. Notice that \ref{eq:duality} plays no substantial role. Accordingly, we can have that
$\ax{RM}$ is valid for permission, if $\PERM$ and $\OBL$ are duals and the logic for $\OBL$ is
a monotonic modal logic, or $\PERM$ is independent of $\OBL$ and $\ax{RM}$ is assumed for $\PERM$. 

In conclusion, if we want not to completely reject the intuition behind \ref{eq:FCP}, we have two non-exclusive options to be explored in order to avoid the Permission Explosion Problem:

\begin{description}
	\item[No-CPL:] abandon $\Lg{CPL}$ and adopt suitable non-classical logical connectives;
	\item[No-RM:] abandon inference rule $\ax{RM}$ (or schema $\ax{M}$) and endorse very weak modal logics (i.e., the classical ones \cite[chap. 8]{Chellas1980}).\footnote{We state in Section \ref{sec:RE} why it is convenient not to drop $\ax{RE}$.}
\end{description}

Our paper aims at exploring under what conditions \textbf{No-CPL} can be avoided by accepting at least a restricted version of \ref{eq:FCP}. Hence, it seems that \textbf{No-RM} thesis must be accepted.

\subsection{Problem 2: Closure under Logical Equivalence Problem}\label{sec:RE}

In the previous section we mentioned that $\ax{RM}$ must be weakened. Hence, we can also drop $\ax{RE}$ and keep axiom schema $\ax{M}$. This choice could look satisfactory for those who consider problematic the fact that the logic for $\PERM$ is closed under logical equivalence.

We take here another route. Incidentally, one can argue that the implausibility of ``It is permitted to eat an apple or not iff it is permitted to sell a house or not'' does not depend on \textbf{RE}, but rather on the fact that ``It is permitted to eat an apple or not'' is $\PERM \top$, which looks quite odd. However, besides this problem---which would lead us to commit to specific philosophical views---dropping $\ax{RE}$ has in general two controversial technical side effects:
\begin{itemize}
	\item it rejects standard semantics for modal logics, since the class of all neighbourhood frames validate \textbf{RE}: \cite{Chellas1980} argued in fact that classical systems (i.e., containing \textbf{RE} but not \textbf{RM}) are the minimal modal logics;
	\item it fails to make, for instance, $\OBL p$ and $\PERM \neg p$ logically incompatible under the Duality Principle (while $\OBL p$ and $\neg \OBL p$ of course are);  similarly,  $\OBL \neg p$ and $\OBL \neg \neg p$, or  $\OBL (p \vee q)$ and $\OBL (\neg p \wedge \neg q)$, are not incompatible too (while they of course should be).
\end{itemize}

In conclusion,  we standardly assume that $\ax{RE}$ holds both for permissions and obligations, which means that any logic for free choice permission must be a \emph{classical system of deontic logic} in \cite{Chellas1980}'s sense, i.e., any modal deontic logic closed under logical equivalence and not under logical consequence.

\subsection{Problem 3: Resource Sensitivity Problem}\label{sec:sensitivity}
It has been noted \cite{Lokhorst} that from ``You may eat an apple or a pear'', one can infer ``You may eat an apple and that You may eat a pear'', but not ``You may eat an apple and a pear'' \cite[][p. 2]{barker:2010:permission}.

We simply observe that the systems proposed in Section \ref{sec:system} do not license in general the inference above. However, a thoughtful treatment of this problem---the Resource Sensitivity Problem---goes beyond the scope of this paper. In fact,
it has been widely discussed in the literature that it is strictly related to considerations from action theory, which have often found solutions shifting from \textbf{CPL} to non-classical logics such as the substructural ones \cite[see, among others,][]{barker:2010:permission,RoyDeon,Dong:2017}.

In conclusion, we do not commit here to find any suitable solution to such a problem.

\section{Three Basic Intuitions}\label{sec:purpose}

We are going to present some deontic systems 
that accommodate restricted variants of \ref{eq:FCP}. This is done under some minimal philosophical assumptions, which can in principle be compatible with several deontic theories. Of course, our approach is not neutral. In this section, we illustrate our fundamental intuitions and assumptions.

\subsection{The Distinction between Norms and Obligations}

We assume in the background a conceptual distinction between norms, on one side, and obligations and permissions, on the other side. The general idea of norms is that they describe conditions under which some behaviours are deemed as `legal'. In the simplest case, a behaviour can be qualified by an obligation (or a prohibition, or a
permission), but often norms additionally specify the consequences
of not complying with them, and what sanctions follow from violations and
whether such sanctions compensate for the violations. The scintilla for this idea is the very influential contribution \cite{alchourron71normative}, which is complementary to the (modal) logic-based approaches to deontic logic. The key feature of this approach is that norms are dyadic constructs connecting applicability conditions to a deontic consequence. A large number of such pairs would constitute an interconnected system
called a {\em normative system} \cite[for more recent proposals in this direction, see][]{Makinson:1999,makinson-torre:2003,ajl:ctd,deon:2016}.

To be clear, \emph{this paper does not present any logic of norms, but our proposal for a logic of obligations and permissions---with restricted variants of \textbf{FCP}---can be
better understood if one keeps in mind some intuitions about how norms should logically behave} and about the relation between the logic of norms and deontic logic. In particular, our assumptions are:
\begin{itemize}
\item obligations and permissions exist because norms generate them when applicable;
\item once obligations and permissions are generated from norms---which requires us to reason about norms---we can still perform some reasoning with the resulting obligations and permissions---this is the task of \emph{deontic logic in a strict sense}, i.e., the logic of obligations and permissions;
\item norms can be in conflict---without being inconsistent--- but this does not hold for obligations and permissions.
\end{itemize}

Hence, we distinguish two levels of analysis: a \emph{norm-logic level} and a resulting \emph{deontic-logic level}. This paper only technically deals with the second level of analysis.

Assume for example that we have two norms $n_1: p \Rightarrow \OBL \neg q$ and $n_2: p \Rightarrow \PERM  q$, where $\Rightarrow$ is any if-then suitable logical relation connecting applicability conditions of norms and their deontic effects. We can indeed have them---for example, in a legal system---but the point is what obligations/permissions we can obtain from them. A rather standard assumption is that in order to correctly derive deontic conclusions we need to solve the conflict between $n_1$ and $n_2$. Specifically, our general view is prudent (or skeptical, as one says in non-monotonic logics), because, unless we know how to solve the conflict (typically, by establishing that $n_1$ is stronger than $n_2$ or vice versa), we do not know if $\OBL \neg q$ or $\PERM q$ holds. Since we do not accept that both can hold, it is pointless to consider at the deontic level that $\OBL \neg q$ \emph{and} $\PERM q$ are true---while any logic of norms can have both $n_1$ and $n_2$.

In conclusion, we impose deontic consistency at the deontic-logic level, i.e., $\OBL p \wedge \OBL \neg p \to \bot$.
%

\subsection{Deontic Incompatibilities, Duality, and \ref{eq:FCP}}

With the above said, the issue is whether \ref{eq:FCP} is an appropriate principle to adopt for normative reasoning. Our view is that this principle in general is not, even when Problem 1 and 2 above are solved. We provide below a simple counterexample to it, which considers the interplay between free choice permissions and prohibitions.

\begin{example}\label{ex:FCP-counterexample} When you have dinner with guests the etiquette allows you to eat or to have a conversation with your fellow guests. However, it is forbidden to speak while eating. 
\end{example}
The full representation of the example is that each choice is
permitted when one refrains from exercising the other one. In a situation
when one eats, there is the prohibition to speak, while when one
speaks, there is the prohibition to eat. Hence, it means that we can 
detach any single permission only if the content of such permission is not 
forbidden. Given that Example~\ref{ex:FCP-counterexample} provides a 
counterexample to \ref{eq:FCP}, the question is whether we want to derive 
the individual permissions when one of the two disjuncts holds and we already 
satisfy the disjunctive permission. The reason is that the individual 
permissions, each on its own, can trigger other obligations or permissions. 
The following example illustrates this scenario.  

\begin{example}\label{ex:return}
  Suppose a shop has the following policy for clothes bought online. If
  the size of an item is not a perfect fit, then the customer is entitled 
  to either exchange the item for free or to keep the item and receive a 10\
  refund. However, customers electing to keep the item are not entitled to the 
  refund, and customers opting for the refund are not entitled to exchange the 
  item for free.  Furthermore, customers who elect to exchange the item (when 
  entitled to do so) have to return it with the original package.  
\end{example}
The example can be formalised as follows:
\begin{gather*}
  \mathit{online}\wedge\neg\mathit{fit}\rightarrow \PERM(\mathit{exchange} \vee \mathit{refund})\\
  \mathit{exchange} \rightarrow \OBL\neg\mathit{refund}\\
  \mathit{refund} \rightarrow \OBL\neg\mathit{exchange}\\
  \PERM\mathit{exchange}\wedge\mathit{exchange} \rightarrow \OBL\mathit{original}
\end{gather*}
Suppose that a customer elects to exchange an item bought online 
that is not a perfect fit instead of asking for the refund. 
Intuitively, given that we cannot derive that exchanging is not forbidden 
($\OBL\neg\mathit{exchange}$) at least the weak permission of exchanging 
the item should hold. However, in a deontic logic without \ref{eq:FCP} 
(or a restricted version of it) we are not able to derive the permission, 
and  then we are not able to derive other obligations or permissions 
depending on it: in the example, the obligation to return the item with 
the original package. 

We will return in Section \ref{sec:FCP} to the logical import of the above 
scenarios in a classical system of deontic logic. For the moment, taking stock 
of the examples we just notice that \ref{eq:FCP} could be reformulated as follows:
\begin{equation}
  \label{eq:fcp-res1}
  (\PERM(p\vee q) \wedge
  (\neg\OBL\neg p \wedge \neg\OBL\neg q))
  \to
  (\PERM p \wedge \PERM q).
\end{equation}
However, assuming \textbf{Duality}, $\neg\OBL\neg p$ is equivalent to
$\PERM p$, thus \eqref{eq:fcp-res1} reduces to
\begin{equation}
  \label{eq:fcp-taut}
  (\PERM(p\vee q) \wedge
  (\PERM p \wedge \PERM q))
  \to
  (\PERM p \wedge \PERM q).
\end{equation}
\eqref{eq:fcp-taut} is a propositional tautology. Thus, \eqref{eq:fcp-res1}
does not extend the expressive power of the logic unless one assumes a logic
where obligation and permission are not the duals.

\subsection{Strong Permission, Classical Systems, and \ref{eq:FCP}}\label{sec:FCP}

When permission is no longer the dual of obligation, we enter the territory of strong permission \cite{vonwright:1963,AB:1981,alchourron-bulygin:1984}\footnote{Besides von Wright's theory \cite{vonwright:1963}, there is another sense in the literature of strong permission \cite{Kamp:1973}. We will briefly return on this in Section \ref{sec:weakest}.}. As is well-known, while it is \emph{sufficient} to show that $\OBL\neg p$ is \emph{not} the case to argue that $p$ is weakly permitted, this does not hold for strong permission, for which the normative system \emph{explicitly} says that there exists at least one norm permitting $p$  \cite[][p. 353--355]{alchourron-bulygin:1984}.

In order to keep track of these two cases at the deontic-logic level, we can standardly distinguish in the deontic language two permission operators, $\wPERM$ for weak permission (such that $\wPERM p \equaldef \neg \OBL \neg p$) and $\sPERM$ for strong permission (where \ref{eq:duality} does not hold).

What is the minimal logic of strong permission at the deontic level in which 
some reasonable version of free choice permission can be accepted? 

We mentioned that \ref{eq:RM} must be rejected. In fact, besides the Permission Explosion Problem, one may also argue that it is reasonable not to derive $\sPERM (p\vee q)$ from any $\sPERM p$ because we could have in the background that the normative system consists just of an explicit norm $a \Rightarrow \sPERM p$. If we have that, in presence of some version of free choice permission, you may also detach $\sPERM q$, which is against the above-mentioned intuition that the strong permission should follow from explicit norms, or from combinations of them in normative systems where all disjuncts are explicitly considered \cite[see, e.g., the discussion in][p. 354--355]{alchourron-bulygin:1984}.

Second, as said above, deontic consistency should be ensured:
\begin{gather}
\OBL p \wedge \sPERM \neg p \to \bot \tag{$\mathbf{D_s}$} \label{eq:Ds}\\
\OBL p \wedge \OBL \neg p \to \bot \tag{$\mathbf{D_w}$} \label{eq:Dw}
\end{gather}
Notice that  \ref{eq:Dw} is the standard $\ax{D}$ axiom of Standard Deontic Logic establishing the so called
\emph{external consistency} of obligations that, in turn, implies consistency among obligations and (weak) permissions.
From \ref{eq:Ds} we obtain, as expected, that strong permission  entails weak permission \cite[see, e.g.,][p. 354]{alchourron-bulygin:1984}, but not the other way around:
\[
\sPERM p \to \wPERM p.
\]
This is reasonable because the fact that at the norm-level we derive that $p$ is permitted using an explicit permissive norm $n$ means that no prohibitive norm $n'$ (forbidding $p$) successfully applies or prevails over $n$.

What about free choice permission? 
Coupling Assumptions 1 and 2 with the distinction between weak and strong permission allows us to identify a guarded variant of \ref{eq:FCP} for strong permission, consisting of two schemata:
\begin{gather}
(\sPERM (p \vee q) \wedge \OBL \neg p) \to \sPERM q \tag{$\ax{AFCP_{\OBL}}$}\label{eq:O}\\
(\sPERM (p \vee q) \wedge \wPERM p \wedge \wPERM q) \to (\sPERM p \wedge \sPERM q) \tag{$\ax{AFCP_{\PERM}}$}\label{eq:P}
\end{gather}
These schemata take stock of what we said: you can detach from a disjunctive strong permission any single strong permission only if this last is weakly permitted.

The idea of the combination of the two axioms is that from repeated applications of \ref{eq:O} and from a
disjunctive permission, we can obtain
the maximal sub-disjunction such that no element is forbidden, and
then, the application of the \ref{eq:P} allows us to derive the individual strong permissions that are not
forbidden.  Notice that we cannot assume the following formula as the axiom for free choice permission.
\[
 \sPERM\bigg(\bigvee_{i=1}^{n}p_i\bigg)\wedge\bigg(\bigwedge_{j=1}^{m<n} \OBL\neg p_j\bigg)\to \bigwedge_{k=m+1}^n \sPERM p_k
\]
The problem is that we do not know in advance how many elements of the disjunctive permission are (individually) forbidden. Consider for example, a theory consisting of the following formulas:
\[
 \sPERM(p \vee q\vee r\vee s \vee t) \qquad \OBL\neg p \qquad \OBL\neg q  \qquad \OBL\neg r
\]
Here, one could use the conjunction $\OBL\neg p\wedge \OBL\neg q$ to obtain $\sPERM r$, $\sPERM s$ and $\sPERM t$, but then we have a contradiction from $\sPERM r$ and $\OBL\neg r$ (from axiom \ref{eq:Ds}). Notice, that in
general, we are not able to use \ref{eq:O} to detach a single (strong) permission, but a disjunction corresponding to the ``remainder'' of the disjunction, that is, in the case above, $\sPERM(s\vee t)$. Then, we can use the \ref{eq:P} to ``lift'' the remaining elements from weak permissions to strong permissions. The only case when we can obtain an individual strong permission from a permissive disjunction is when the remainder is a singleton;  but this means, that all the other elements of the permissive disjunction were forbidden. This further means that a disjunctive strong permission holds if at least one of its elements can be legally exercised. Going back to the example, if one extends the theory with $\OBL\neg s$, then we can derive $\sPERM t$. 

Let us consider again the situation described in Example~\ref{ex:FCP-counterexample}. The scenario can be formalised as follows (where $e$ and $s$ stand, respectively for ``to eat'' and ``to speak''):
\begin{gather*}
  \sPERM(e \vee s)\\
  s \rightarrow \OBL\neg e\\
  e \rightarrow \OBL\neg s
\end{gather*}
In a logic endorsing the unrestricted version of free choice permission, we have $\sPERM e$ and $\sPERM s$. This means that as soon as one exercises one of the choices, we get that the other choice is at the same time permitted and forbidden, a situation that is either paradoxical or contradictory. Thus, the only way to avoid
this kind of conflict is to refrain from exercising any of the two choices. However, this means that one is
not really free to choose between the two options. Accordingly, either one has to adopt a restricted version of the free choice permission or abandon it. Notice, that axiom \ref{eq:O} allows us to conclude that given $e$, $s$ is forbidden ($\OBL\neg s$), and thus that $e$ is permitted ($\sPERM e$); similarly, one gets $\sPERM s$ from $s$, which implies $\OBL\neg e$.
Similarly, for Example~\ref{ex:return} when we formalise it using strong permission $\sPERM$ instead of $\PERM$, Axiom \ref{eq:O} allows us to derive $\sPERM\mathit{exchange}$ from which we can conclude $\OBL\mathit{original}$.

Consider $\ax{AFCP_{\PERM}}$. One may argue why, in symmetry with $\ax{AFCP_{\OBL}}$, we cannot rather have
\begin{equation}\tag{$\ax{AFCP2_{\PERM}}$}
	\label{eq:P2}
(\sPERM (p \vee q) \wedge \wPERM p ) \to \sPERM p
\end{equation}
Technically, it is obvious that $\ax{AFCP2_{\PERM}}$ implies $\ax{AFCP_{\PERM}}$ but not the other way around, so both options are available. The variant $\ax{AFCP_{\PERM}}$ is more prudent in that it licenses the detachment of an individual strong permission \emph{only if} the normative system explicitly deals with that specific disjunct, while the second allows for the derivation in a slightly more relaxed way. So, if one wants to strictly reframe the structure of standard $\ax{FCP}$ in a guarded version but does not want $\ax{AFCP2_{\PERM}}$, then $\ax{AFCP_{\PERM}}$ is the right option.

We should notice that the above schemata for free choice permission do not necessarily require the technical idea of deontic consistency, unless we assume---but we don't---that obligation implies strong permission, and despite the fact that the consistency problem can occur if we endorse \ref{eq:Ds}---as we do--- and so that strong permission implies weak permission.
In fact, if we do not validate \ref{eq:RM}, we would need anyway to model the following scenario
\begin{equation}\label{eq:case}
(\sPERM (p \vee q) \wedge \OBL (\neg p \wedge r)) \to \sPERM q
\end{equation}
as a true instance of our intuition, despite the fact that $\OBL (\neg p \wedge r)$ and $\sPERM p$ are not inconsistent.
%
One may rather argue that any state exercising the permission $\sPERM p$ --i.e., a state where $p$ holds-- does not comply with $\OBL (\neg p \wedge r)$. If so, the condition that detached permissions are compliant with obligations is a \emph{rational requirement for free choice permission}, which is not technically needed in classical modal systems for ensuring standard logical consistency between deontic statements.\footnote{Indeed, it is a trivial result in classical systems that the inference of $\OBL \neg p$ from $\OBL (\neg p \wedge r)$ is not in general valid. Semantically, it is also immediate to build a neighbourhood model which falsifies that inference or which, in a similar perspective, admits the truth of both $\OBL (\neg p \wedge r)$ and $\sPERM p$. (This holds if $\sPERM$ is treated as an independent $\Box$ operator with no further conditions.)}

Therefore, if one wants to consider (\ref{eq:case}) as an instance of free choice permission in logics that do not satisfy $\ax{RM}$, we need to replace \ref{eq:O} and \ref{eq:P} with the following inference rules:
\begin{equation}\tag{$\mathbf{IFCP_{\OBL}}$}
	\label{eq:RFCPO}
  \frac{\sPERM (p \vee q) \wedge \OBL r \qquad \vdash r \to \neg p}{\sPERM q}
\end{equation}
\begin{equation}\tag{$\mathbf{IFCP_{\PERM}}$}
	\label{eq:RFCPP}
  \frac{\sPERM (p \vee q) \wedge (\wPERM r \wedge \wPERM s)\qquad \vdash r \to p \qquad \vdash s \to q}{\sPERM p \wedge \sPERM q}.
\end{equation}
Of course, the same remark we made for $\ax{AFCP_{\PERM}}$ and $\ax{AFCP2_{\PERM}}$ holds here, too, so we may have the following alternative:
\begin{equation}\tag{$\mathbf{IFCP2_{\PERM}}$}
	\label{eq:RFCPP2}
  \frac{\sPERM (p \vee q) \wedge \wPERM r \qquad \vdash r \to p }{\sPERM p}.
\end{equation}

In the discussion leading to the formulation of the inference rules above, we provided the intuition and an example for \ref{eq:RFCPO}. Let us now examine a few interesting cases for \ref{eq:RFCPP} and \ref{eq:RFCPP2}.
For the first case, we consider the following instance
\[
 \sPERM(p\vee q), \quad \wPERM(p\wedge r)
\]
In a logic with $\ax{RM}$ we can derive $\wPERM p$ from $\wPERM(p\wedge r)$, since $(p\wedge r)\to p$ is a tautology. The inference rules \ref{eq:RFCPP2} and \ref{eq:RFCPP} allow us to replicate this type of reasoning without being forced to derive $\wPERM p$, thus we can apply \ref{eq:RFCPP2} to the instance above to obtain $\sPERM p$.

The second situation  is given by the following theory:
\[
  \sPERM(p \vee q\vee r), \quad \wPERM p, \quad \wPERM q, \quad \OBL\neg r
\]
Since $p\to p \vee r$ and $q\to q\vee r$ are tautologies, the theory implies, by applications of \ref{eq:RFCPP2}, $\sPERM p$, $\sPERM q$, $\sPERM(p\vee r)$ and $\sPERM(q\vee r)$.

It is easy to verify that $\ax{RM}$ in conjunction with any of the axiom \ref{eq:O}, \ref{eq:P} and \ref{eq:P2} makes the corresponding inference rules admissible. But, as we have seen in Section~\ref{sec:FCP}, the combination of \ref{eq:FCP} and $\ax{RM}$ leads to the Permission Explosion Problem.  This is not the case when we replace \ref{eq:FCP} with \ref{eq:P} (or \ref{eq:P2} or the corresponding inference rules), since the guarded version of free choice permission allows us to derive an individual strong permission from a disjunctive (strong) permission containing the individual element only when we have that the individual permission is also weakly permitted. In general, the axiom and
inference schemata we have proposed do not suffer from the permission explosion, though  the resulting logics
have a ``controlled'' permission explosion in the sense that every weak permission is lifted to a strong permission in case a tautology is explicitly strongly permitted, as the following derivation shows:
\[
\begin{array}{lll}
  1. & \sPERM(p\vee\neg p) & Hyp.\\
  2. & \wPERM q  & Hyp.\\
  3. & \sPERM(q\vee p\vee\neg p) & 1,\ \ax{RE}\\
  4. & \sPERM q  & 1,2,3,\ \ax{AFCP_P}
\end{array}
\]
The consequence of the controlled permission explosion is that whenever the logic for the strong permission is a normal modality, the notion of strong permission collapses to that of weak permission.\todo{\tiny vogliamo dire qualcosa sulla concezione di Alchurron e Bulygin che, logicamente, c'e' un solo tipo di permesso e la differenza e' nell'uso?
NINO: secondo me è meglio di no perche' noi assumiamo la differenza nelle modalita'. Non  mi è chiara la ragione che lega le condizioni di collasso discusse qui e quanto detto da Alchourron e Bulygin. O meglio, la cosa e' chiara---semanticamente normal modality significa perdere la distinzione tra norme (una sola relazione di accessibilita'), ma FCP è irrilevante, cio' che contna è che siamo in un sistema normale. Io eviterei casini, a meno che tu abbia le idee piu' chiare di me.\\}
To avoid this issue one could either impose the axiom $\neg\sPERM\top$\todo{Inserire riferimenti? discussione?} or reject $\ax{RE}$; however, the later seems to be more problematic (see the discussion in Section~\ref{sec:RE} or part of the discussion in Section~\ref{sub:asher-bonevac} below).

\section{Two Related Works}\label{sec:related_work}
In this section we review in some detail two related works that have directly implications with respect to our proposal \cite{Asher2005,anglberger_gratzl_roy_2015}. In fact, even though they have a different philosophical backgrounds, they propose \emph{simple non-normal axiomatisations for obligation and permission}---as we do---which avoid, e.g., Problem 1 and which are based on the concept of free choice permission as strong permission or, anyway, as a type of permission without \ref{eq:duality}.


\subsection{Asher and Bonevac's Analysis}
\label{sub:asher-bonevac}

Asher and Bonevac \cite{Asher2005}'s analysis presents a deontic logic based on  Anderson-Kanger reduction of obligations and permissions. Assume $S$ denotes as usual Sanction. Then, $\OBL a =_{\mathit{def}} \Box (a \to S)$---where $\Box$ is an $\Lg{S5}$ operator---is dropped by introducing a suitable and weak conditional logic for $>$, such that
\begin{gather*}
\OBL a\equaldef (\neg a > S) \qquad \qquad \sPERM a \equaldef (a > \neg S)
\end{gather*}
Obligations and permissions are defined through the concept of sanction, but \ref{eq:duality} no longer holds, since Asher and Bonevac \cite{Asher2005} argue that free choice permission is strong permission.

The logic for $>$ is not closed under logical equivalence in conditional antecedents, so the resulting $\sPERM$ and $\OBL$ are not closed under logical equivalence. The overall system is non-monotonic and can naturally handle cases such as the ones expressed in \ref{eq:O}.

So, we share with \cite{Asher2005} important assumptions. However, Asher and Bonevac's proposal suffers from some drawbacks that we consider difficult to accept. Besides the fact that the closure under logical equivalence does not hold (see Section \ref{sec:RE}), consider the following scenario \cite[p. 311]{Asher2005}:

\begin{example}[(Soup-Eggroll scenario \cite{Asher2005})]\label{ex:soup-eggroll}
Assume that $OK \equaldef \neg S$.
\begin{quote}
``Suppose we go to the Chinese restaurant. There it’s part of the context that while you may have soup or eggroll, you can’t have both. Thus,
\[
\neg (A > \mathit{OK}) \vee \neg (B > \mathit{OK})
\]
holds. But while we now can derive defeasibly that you may have, e.g., an eggroll, you can't have both eggroll and soup, at least without paying extra.''
\end{quote}
\end{example}

The Soup-Eggroll scenario, as analysed by \cite{Asher2005}, is debatable. Indeed, the fact that you can't have both soup and eggroll means that having both is forbidden. The point is that a prohibition amounts here to the fact that both strong permissions are false rather that having
$\OBL \neg (\mathit{soup} \wedge \mathit{eggroll})$, namely,
$\OBL (\neg \mathit{soup} \vee \neg \mathit{eggroll})$.
This is due to the fact that, in the example, prohibition is the negation of strong permission and not the obligation of the opposite.

Consider the following example. In some card games, it is obligatory not to play a trump card when one is the first to play; this can be expressed as when one is the first to play it is forbidden to play a trump card. Intuitively those two statements seem to be equivalent, but their formal representation, namely
\[
 (\neg\neg \mathit{trump} > S)
 \qquad
 \neg(\mathit{trump} > \neg S)
\]
are not. Notice, also, that, since the first statement contains a negation and the logic for $>$ is not closed
under logical equivalence, the proper translation is with the double negation, and as a consequence,  $(\neg\neg a > S)$ and $(a > S)$ are not equivalent; while this might be acceptable in deontic logics based on non-involutive multi-valued logics, it seems counterintuitive when the underlying logic is the bivalent classical propositional logic.

%

In conclusion, while \cite{Asher2005} has the great merit of identifying fundamental intuitions for free choice permission, it fails to frame those intuitions in a convincing general theory of obligations.

\subsection{Open Reading of Permission and Obligation as Weakest Permission}\label{sec:weakest}

A more recent proposal discussing free choice permission is due to \cite{anglberger_gratzl_roy_2015}, though the philosophical background of this work is significantly different from ours.

First of all, \cite{anglberger_gratzl_roy_2015} works on the concept of open reading of permissions \cite{Lewis:1979,Broersen:2004}. Consider that ``it is permitted to board the plane''---$\PERM p$---and assume that $p$ is an action type, i.e., something saying that ``there are many, mutually exclusive action tokens of that type''. Hence, there ``might be many ways to board a plane. There might be more than one gate to go through, there might be several times within a fixed period when one can proceed, etc.'', i.e., there are many action tokens \cite[][p. 808]{anglberger_gratzl_roy_2015}. How to read $\PERM p$ \cite{Lewis:1979}? The authors adopt the so-called open reading, according to which at least one token of type $p$ (but possibly not all) is acceptable according to the normative system.\footnote{Permission defined through the open reading is sometimes called strong permission \cite{Kamp:1973}, which is different from \cite{vonwright:1963}'s notion \cite[see][p. 808, fn. 1]{anglberger_gratzl_roy_2015}. We can ignore this issue here, by generically assuming that strong permission is just not the dual of obligation.}

Hence, rational obligations and permissions should be seen, respectively, as giving \emph{necessary} and \emph{sufficient} conditions for rational agents (typically, in game-theoretic settings) \cite[][sec. 2.3]{anglberger_gratzl_roy_2015}. An action type $p$ is obligatory whenever it is exactly the normatively ideal action type: \emph{obligation as weakest permission}. This means that ``playing any action type that rules out being rational is forbidden. This is not the case in the logic of obligations as weakest permissions. There the unique obligation bearing on the players is to play a rational strategy'' \cite[][p. 17]{Dong:2017}.

We believe that this contribution may have a general import on the debate on permission, which goes beyond the philosophical discussion on the open reading.

Technically, the modal system of \cite{anglberger_gratzl_roy_2015} is classical, too, in \cite{Chellas1980}'s sense, and $\PERM$ is not the dual of $\OBL$, but an independent $\Box$-operator.

Conceptually, it is interesting in general to explore a deontic logic that allows for deriving a unique obligation \cite[][p. 817]{anglberger_gratzl_roy_2015}. If we assume the distinction between norms and obligations, then this means that, given a normative system $\mathbf{N}$ whose norms prescribe $a_1 , \dots , a_n$, then we can only have in the deontic logic $\OBL (a_1 \wedge \dots \wedge a_n)$ but not, for example, any obligation $\OBL a_k$ ($1\leq k \leq n$). This means two things: (a) $\ax{M}: \OBL(a \wedge b) \to \OBL a \wedge \OBL b$ is not valid thus leading to rejecting $\ax{RM}$, (b) we exclude the possibility of having any norm supporting some other conclusion, which is made applicable by one single obligation.

As for point (b), while in several cases it is not harmful to only derive the entire obligative conjunctive content of a normative system---this is not a problem for checking compliance, for instance---there are cases where one may need to logically speak of single obligations. Consider a normative system, which states that some obligations are conjunctively typical only of a certain type of entities---as many legal systems do. For instance, assume that $\OBL (a_1 \wedge \dots \wedge a_n)$ holds only for specific commercial entities, such as corporations. This means that
\[
\OBL (a_1 \wedge \dots \wedge a_n) \equiv \mathit{Corporation}
\]
However, this cannot exclude that a subset of those obligations (e.g., $a_1$ and $a_2$) implies that your company is a partnership:
\[
\OBL (a_1 \wedge a_2) \to \mathit{Partnership}.
\]
That this is not admitted in a deontic logic looks to us too restrictive if we go beyond the domain of the open reading of permission.

\section{Six Minimal Deontic Axiomatisations with Guarded Free Choice Permission}\label{sec:system}

Finally, we present some minimal deontic systems, six Hilbert-style deontic systems admitting a guarded version of \ref{eq:FCP}. The systems that we present are not too weak from the inferential viewpoint, as far as permission is concerned, and do not commit to weakening any specific logic for obligations.

\subsection{Language, Axioms and Inference Rules}\label{sec:axioms}

The modal language and the concept of well formed formula are defined as usual \cite[see][]{Chellas1980,Blackburn2001}. We just recall that we have three modal operators, two $\Box$ operators, $\OBL$ for obligations and $\sPERM$ for strong permissions, and $\wPERM$ for weak permission. As usual, we assume $\wPERM$ to be an abbreviation for $\neg \OBL \neg$.

For convenience, let us synoptically recall below all relevant schemata and inference rules, where $\Box \in \set{\OBL, \sPERM}$.

\paragraph{Inference Rules:}

	\begin{description}
		\item[$\ax{RE}:=$] $ \vdash A \equiv B ~ \Rightarrow ~ \vdash \Box A \leftrightarrow \Box B$
		\item[$\ax{RM}:=$]$\vdash A \to B ~ \Rightarrow ~ \vdash \Box A \to \Box B$
		\item[$\ax{IFCP_{\OBL}}:=$] $\sPERM (p \vee q) \wedge \OBL r$ and $\vdash r \to \neg p ~ \Rightarrow ~ \sPERM q$
		\item[$\ax{IFCP_{\PERM}}:=$] $\sPERM (p \vee q) \wedge (\wPERM r \wedge \wPERM s)$, $\vdash r \to p$ and $\vdash s \to q ~ \Rightarrow ~ \sPERM p \wedge \sPERM q$
		\item[$\ax{IFCP2_{\PERM}}:=$] $\sPERM (p \vee q) \wedge \wPERM r$, $\vdash r \to p ~ \Rightarrow ~ \sPERM p$.
	\end{description}

\paragraph{Schemata:}
	\begin{description}
		\item[$\ax{M}:=$] $\Box(p \wedge q) \to (\Box p \wedge \Box q)$
		\item[$\ax{AFCP_{\OBL}}:=$] $(\sPERM (p \vee q) \wedge \OBL \neg p) \to \sPERM q$
		\item[$\ax{AFCP_{\PERM}}:=$] $(\sPERM (p \vee q) \wedge \wPERM p \wedge \wPERM q) \to (\sPERM p \wedge \sPERM q)$
		\item[$\ax{AFCP2_{\PERM}}:=$] $(\sPERM (p \vee q) \wedge \wPERM p) \to \sPERM p$
		\item[$\ax{D_s}:=$] $\OBL p \wedge \sPERM \neg p \to \bot $
		\item[$\ax{D_w}:=$] $\OBL p \wedge \wPERM \neg p \to \bot$
		\item[$\ax{P_sP_w}:=$] $\sPERM p \to \wPERM p$.
	\end{description}


Given the discussion of Section \ref{sec:purpose}, we can identify some deontic systems, as specified in Table \ref{table:systems}. Notice that we consider also systems $\Lg{FCP3}$ and $\Lg{FCP6}$, which are  monotonic, so they contain $\ax{RM}$. Strictly speaking, this is the limit which we cannot trespass, since we have three restricted forms of Permission Explosion. We will return on this in the concluding section of the paper.

{\footnotesize
\begin{table}
	\begin{tabular}{| l | c | r |}
	\hline
	Deontic System  &  Properties  & Derivable \\
	\hline
	\hline
	$\Lg{E} := \ax{RE}$  &   &  \\
	\hline
		$\Lg{Min} := \ax{RE}\oplus \ax{D_s} \oplus \ax{D_w}$  &   &  $\ax{P_sP_w}$\\
		\hline
	$\Lg{FCP_1} := \ax{Min}\oplus \ax{IFCP_{\OBL}} \oplus \ax{IFCP_{\PERM}}$  &   & $\ax{P_sP_w}$ \\
   &   &  $\ax{AFCP_{\OBL}}$, $\ax{AFCP_{\PERM}}$\\
	\hline
	$\Lg{FCP_2} := \ax{Min}\oplus \ax{AFCP_{\OBL}} \oplus \ax{AFCP_{\PERM}}$  &  $\Lg{FCP_2}\subset \Lg{FCP_1}$ & $\ax{P_sP_w}$ \\
	\hline
	$\Lg{FCP_3} := \Lg{FCP_2}\oplus\ax{M}$  &  $\Lg{FCP_1}\subset \Lg{FCP_3}$ & $\ax{P_sP_w}$ \\
	  &  & $\ax{IFCP_{\OBL}}, \ax{IFCP_{\PERM}}$ \\
  	\hline
  	$\Lg{FCP_4} := \Lg{Min}\oplus\ax{AFCP_{\OBL}}\oplus\ax{AFCP2_{\PERM}}$  &  $\Lg{FCP_2}\subset \Lg{FCP_4}$ & $\ax{P_sP_w}$, $\ax{AFCP_{\PERM}}$ \\
  	\hline
  	$\Lg{FCP_5} := \Lg{Min}\oplus\ax{IFCP_{\OBL}}\oplus\ax{IFCP2_{\PERM}}$  &  $\Lg{FCP_1}\subset \Lg{FCP_5}$ & $\ax{P_sP_w}$, $\ax{IFCP_{\PERM}}$ \\
	& $\Lg{FCP_4}\subset\Lg{FCP_5}$ & $\ax{AFCP_{\OBL}}$, $\ax{AFCP_{\PERM}}$\\
	& & $\ax{AFCP2_{\PERM}}$\\
	\hline
  	$\Lg{FCP_6} := \Lg{FCP_4}\oplus\ax{M}$  &  $\Lg{FCP_3}\subset \Lg{FCP_6}$ & $\ax{P_sP_w}$, $\ax{IFCP_{\PERM}}$ \\
	& $\Lg{FCP_5}\subset\Lg{FCP_6}$ & $\ax{AFCP_{\OBL}}$, $\ax{AFCP_{\PERM}}$\\
	& & $\ax{AFCP2_{\PERM}}$, $\ax{IFCP2_{\PERM}}$\\
	& & $\ax{IFCP_{\OBL}}$, $\ax{IFCP2_{\OBL}}$\\
	\hline
	\end{tabular}
	\label{table:systems}
	\caption{Deontic Systems}
\end{table}
}

\subsection{Semantics and System Properties}\label{sec:semantics}

Let us begin with standard concepts. Assume that $\prop$ is
the set of atomic sentences.

\begin{definition}\label{def:neighbourhood-frame}
\sloppy A \emph{deontic neighbourhood frame} $\cF$ is a structure
$\tuple{W, \cNO , \cNP}$ where
  \begin{itemize}
  \item $W$ is a non-empty set of possible worlds;
  \item $\cNO$ and $\cNP$ are functions $W\mapsto 2^{2^W}$.
  \end{itemize}
\end{definition}

\begin{definition}\label{def:neighbourhood}
\sloppy  A \emph{deontic neighbourhood model} $\cM$ is a structure
  $\tuple{W, \cNO , \cNP , V}$ where $\tuple{W,
  \cNO , \cNP}$ is a deontic neighbourhood frame and
  $V$ is an evaluation function $\prop\mapsto 2^W$.
\end{definition}

\begin{definition}[Truth in a model]\label{def:neigh-truth}
 Let $\cM$ be a model $\tuple{W, \cNO , \cNP , V}$ and
 $w\in W$. The truth of any formula $p$ in $\cM$ is defined inductively
 as follows:
\begin{enumerate}
\item standard valuation conditions for the boolean connectives;
\item $\cM,w\models \OBL p$ iff $\tset{p}\in
  \cNO (w)$,
\item $\cM,w\models \sPERM p$ iff $\tset{p}\in
  \cNP (w)$,
\item $\cM,w\models \wPERM p$ iff $W- \tset{p}\not\in
  \cNO (w)$,
\end{enumerate}
where, as usual, $\tset{p}$ is the truth set of $p$ wrt to $\cM$:\footnote{Whenever clear from
    the context we drop the references to the model.}
\[
\tset{p} = \set{w\in W: \cM ,w\models p}.
\]
\end{definition}
A formula $p$ is \emph{true at a world} in a model
iff $\cM,w\models p$; \emph{true in a model} $\cM$, written $\cM\models
p$ iff for all worlds $w\in W$, $\cM,w\models p$; \emph{valid in a
  frame} $\cF$, written $\cF\models p$ iff it is true in all models
based on that frame; \emph{valid in a class $\cC$ of frames}, written $\cC\models p$, iff it is valid in all frames in the class. Analogously, an inference rule $P_1 , \dots P_n \Rightarrow C$ (where $P_1 , \dots P_n$ are the premises and $C$ the conclusion) is valid in a class $\cC$ of frames iff, for any $\cF\in\cC$, if $\cF\models P_1 , \dots , \cF\models P_n$ then $\cF\models C$\footnote{Of course, if any $P_k$ has the form $\vdash p$ then $\cF\models P_1$ trivially means $\cF\models p$.}.

We can now characterise different classes of deontic neighbourhood frames that are adequate of the deontic systems in Table \ref{table:systems}.
\begin{definition}[Frame Properties]\label{def:frame-properties}
Let $\cF = \tuple{W, \cNO , \cNP}$ be a deontic neighbourhood frame.
	\begin{itemize}
	\item \textbf{$\Box$-\emph{supplementation}}: $\cF$ is $\Box$-\emph{supplemented}, $\Box\in\set{\OBL, \PERM}$, iff for any $w\in W$ and $X,Y\subseteq W$, $X\cap Y\in \cN_{\Box} (w) \Rightarrow X\in \cN_{\Box} (w)\, \& \, Y\in \cN_{\Box} (w)$;
	\item \textbf{$\wPERM$-\emph{coherence}}: $\cF$ is $\wPERM$-\emph{coherent} iff for any $w\in W$ and $X\subseteq W$,
	$X\in \cNO (w) \Rightarrow W-X\not\in \cNO (w)$;
	\item \textbf{$\sPERM$-\emph{coherence}}: $\cF$ is $\sPERM$-\emph{coherent} iff for any $w\in W$ and $X\subseteq W$,
	$X\in \cNP (w) \Rightarrow W-X\not\in \cNO (w)$;
	\item \textbf{$\mathbf{AFCP_{\OBL}}$-\emph{permission}}: $\cF$ is $\mathbf{AFCP_{\OBL}}$-\emph{permitted} iff for any $w\in W$ and
	$X,Y\subseteq W$, $X\cup Y\in \cNP (w) \, \& \, W-Y \in \cNO(w)\Rightarrow X\in \cNP (w)$;
	\item \textbf{$\mathbf{AFCP_{\PERM}}$-\emph{permission}}: $\cF$ is $\mathbf{AFCP_{\PERM}}$-\emph{permitted} iff for any $w\in W$ and
	$X,Y\subseteq W$, $X\cup Y\in \cNP (w) \, \& \, W-X \not\in \cNO(w) \, \& \, W-Y \not\in \cNO(w) \Rightarrow X\in \cNP (w)\, \& \, X\in \cNP (w)$;
	\item \textbf{$\mathbf{AFCP2_{\PERM}}$-\emph{permission}}: $\cF$ is $\mathbf{AFCP2_{\PERM}}$-\emph{permitted} iff for any $w\in W$ and
	$X,Y\subseteq W$, $X\cup Y\in \cNP (w) \, \& \, W-X \not\in \cNO(w) \Rightarrow X\in \cNP (w)$;
	\item \textbf{$\mathbf{IFCP_{\OBL}}$-\emph{permission}}: $\cF$ is $\mathbf{IFCP_{\OBL}}$-\emph{permitted} iff for any $w\in W$ and
	$X,Y,Z\subseteq W$, $X\cup Y\in \cNP (w) \, \& \, Z\subseteq (W-Y)\, \& \, W-Z \in \cNO(w)\Rightarrow X\in \cNP (w)$;
	\item\sloppy \textbf{$\mathbf{IFCP_{\PERM}}$-\emph{permission}}: $\cF$ is $\mathbf{IFCP_{\PERM}}$-\emph{permitted} iff for any $w\in W$ and
	$X,Y,Z,Q\subseteq W$, $X\cup Y\in \cNP (w) \, \& \, Z\subseteq X \, \& \, Q\subseteq Y
	\, \& \, W-Z \not\in \cNO(w) \, \& \, W-Q \not\in \cNO(w) \Rightarrow X\in \cNP (w)\, \& \, Y\in \cNP (w)$;
	\item \textbf{$\mathbf{IFCP_{\PERM}}$-\emph{permission}}: $\cF$ is $\mathbf{IFCP2_{\PERM}}$-\emph{permitted} iff for any $w\in W$ and
	$X,Y,Z\subseteq W$, $X\cup Y\in \cNP (w) \, \& \, Z\subseteq X 	\, \& \, W-Z \not\in \cNO(w)   \Rightarrow X\in \cNP (w)$.
	\end{itemize}
\end{definition}

Here below are some relevant characterisation results. The proofs are in the Appendix.


\begin{restatable}{lemma}{characterisation}
\label{lemma:characterisation_results}
For any deontic neighbourhood frame $\cF$,
\begin{enumerate}

  \item\label{en:Ds} $\ax{D_s}$ is valid in the class of $\sPERM$-\emph{coherent} frames;
  \item\label{en:Dw} $\ax{D_w}$ is valid in the class of $\wPERM$-\emph{coherent} frames;
  \item\label{en:AOBL} $\ax{AFCP_{\OBL}}$ is valid in the class of $\mathbf{AFCP_{\OBL}}$-\emph{permitted} frames;
  \item\label{en:APERM} $\ax{AFCP_{\PERM}}$ is valid in the class of  $\mathbf{AFCP_{\PERM}}$-\emph{permitted} frames;
  \item\label{en:APERM2} $\ax{AFCP2_{\PERM}}$ is valid in the class of  $\mathbf{AFCP2_{\PERM}}$-\emph{permitted} frames;
  \item\label{en:IOBL} $\ax{IFCP_{\OBL}}$ is valid in the class of  $\mathbf{IFCP_{\OBL}}$-\emph{permitted} frames;
  \item\label{en:IPERM} $\ax{IFCP_{\PERM}}$ is valid in the class of $\mathbf{IFCP_{\PERM}}$-\emph{permitted} frames;

  \item\label{en:IPERM2} $\ax{IFCP2_{\PERM}}$ is valid in the class of $\mathbf{IFCP2_{\PERM}}$-\emph{permitted} frames.
\end{enumerate}
\end{restatable}

Completeness results for the three deontic systems are ensured: again see the Appendix for a proof.

\begin{restatable}{theorem}{completeness}
\label{th:completeness}
\
\begin{enumerate}[\indent (a)]
	\item the system $\Lg{E}$ is sound and complete w.r.t. the class of deontic neighbourhood frames;
	\item the system $\Lg{Min}$ is sound and complete w.r.t. the class of $\sPERM$- and $\wPERM$-\emph{coherent} frames;
	\item the system $\Lg{FCP_1}$ is sound and complete w.r.t. the class of $\mathbf{IFCP_{\OBL}}$- and $\mathbf{IFCP_{\PERM}}$-\emph{permitted} frames;
	\item the system $\Lg{FCP_2}$ is sound and complete w.r.t. the class of $\mathbf{AFCP_{\OBL}}$- and $\mathbf{AFCP_{\PERM}}$-\emph{permitted} frames;
	\item the system $\Lg{FCP_3}$ is sound and complete w.r.t. the class of $\PERM$-\emph{supplemented}, $\mathbf{AFCP_{\OBL}}$- and $\mathbf{AFCP_{\PERM}}$-\emph{permitted} frames;
		\item the system $\Lg{FCP_4}$ is sound and complete w.r.t. the class of $\mathbf{AFCP_{\OBL}}$- and $\mathbf{AFCP2_{\PERM}}$-\emph{permitted} frames;
		\item the system $\Lg{FCP_5}$ is sound and complete w.r.t. the class of $\mathbf{IFCP_{\OBL}}$- and $\mathbf{IFCP2_{\PERM}}$-\emph{permitted} frames;
		\item the system $\Lg{FCP_6}$ is sound and complete w.r.t. the class of $\PERM$-\emph{supplemented}, $\mathbf{IFCP_{\OBL}}$- and $\mathbf{IFCP2_{\PERM}}$-\emph{permitted} frames.
\end{enumerate}
\end{restatable}

Finally, a corollary showing the relative strength of the six deontic systems.

\begin{restatable}{corollary}{inclusion}
\label{cor:inclusion}
	\
\begin{enumerate}[\indent (i)]
	\item\label{cor:1} $\Lg{FCP_2}\subset \Lg{FCP_1} \subset \Lg{FCP_3}\subset \Lg{FCP_6}$,\\
	    $\Lg{FCP_2}\subset \Lg{FCP_4} \subset \Lg{FCP_5} \subset \Lg{FCP_6}$ and\\
	    $\Lg{FCP_1}\subset \Lg{FCP_5}$.
	\item\label{cor:2} Let $\Lg{L_1},\Lg{L_2}\in\set{\Lg{FCP_i},1\leq i\leq 6}$, and let $\cC_1$ and $\cC_2$
	  be classes of frames adequate for $\Lg{L_1}$ and $\Lg{L_2}$. If $\Lg{L_1}\subset\Lg{L_2}$ then $\cC_2\subset\cC_1$.
\end{enumerate}
\end{restatable}

\section{Conclusions}\label{sec:conclusions}

In this paper we have investigated how, and if the notion of free choice permission is admissible in modal deontic logic. As is well known, several problems can be put forward in regard to this notion, the most fundamental of them being the so-called Permission Explosion Problem, according to which all systems containing \ref{eq:FCP} and closed under \ref{eq:RM} and \ref{eq:RM-P} license the derivation of any arbitrary permission whenever at least one specific permission is true. 

We argued (Section~\ref{sec:explosion}) that a plausible solution to this problem is to jump from monotonic into  classical deontic logics, i.e., systems closed under \ref{eq:RE} but not \ref{eq:RM}. This solution does not necessarily mean that the resulting deontic system is very weak, as far as permission is concerned, if further schemata and inference rules are added (Sections~\ref{sec:FCP} and \ref{sec:semantics}). 

The basic intuitions for extending classical deontic logics are the following: 
\begin{enumerate}
	\item We assume in background the distinction between norms and obligations/permissions. While we conceptually accept that the normative system may contain conflicting norms, it is logically inadmissible that such norms generate actual conflicting obligations/permissions since conflicts must be rationally solved, otherwise no obligation/permission can be obtained; hence, we validate schemata \ref{eq:Ds} and \ref{eq:Dw};
	\item Free choice permission is strong permission, meaning that it is a permission generated by explicit permissive norms;
	\item The possibility of detaching single strong permissions from disjunctive strong permissions, i.e., $\sPERM q$ from $\sPERM (p \vee q)$ strictly depends on the fact that $\OBL \neg p$ is not the case.
\end{enumerate}
Taking the above points into account, we thus proposed different guarded variants of \ref{eq:FCP} that significantly increase the inferential power of the logic. In particular, six Hilbert-style classical deontic systems were presented.

We observed that four of these systems are classical modal systems, while we can have other two acceptable systems which are monotonic. In fact, the fact that those two systems are closed under \ref{eq:RM} does not lead to full Permission Explosion, but only to a  ``controlled'' version of it: indeed, in systems like $\Lg{FCP3}$ any permission is obtainable via free choice permission \emph{only if} it is not incompatible with existing prohibitions.

Some directions for future work can be identified. In particular:
\begin{itemize}
	\item It is still an open issue to fully discuss the Resource Sensitivity Problem in our setting. In fact, while we argued that this problem goes beyond our paper, there are scenarios where our intuitions are relevant for this problem as well. For example, suppose that there is a fruit basket in the kitchen containing a banana and an apple. Bob and Alice are permitted to eat the banana or the apple and Alice first eats the former. Bob cannot do anything but take the apple. However, if Bob is allergic of apples, so no permission can be reasonably derived because it is forbidden for him to eat the apple.   
	\item Our idea of free choice permission relies on the fact that no strong permission can be detached from a disjunctive permissive expression if another norm allows for deriving a conflicting obligation. Hence a full understanding of schemata such as $\ax{AFCP_{\OBL}}$ or $\ax{AFCP_{\PERM}}$ may benefit for an explicit logical treatment of the logic of norms adopting defeasible reasoning \cite{Dong:2017}.
\end{itemize}
%


%
%
%

\subsection*{Acknowledgments}
This work was partially supported by the EU H2020 research and innovation programme under the Marie Sk{\l}odowska-Curie grant agreement No.\@ 690974 for the project MIREL: \emph{MIning and REasoning with Legal texts}.

\bibliographystyle{plainnat}
\bibliography{modal_logic,normative_mas}

\appendix
\section{Basic Properties of the Deontic Systems}

Let us start by proving Lemma \ref{lemma:characterisation_results}.

\characterisation*
\begin{proof}
The proof for case (\ref{en:Ds}) is straightforward. The proof of (\ref{en:Dw}) is trivial and standard. Both are omitted.
\todo{Allora, dire che una formula è valida in una classe di frames, significa dire che per ogni mondo, e per tutti i modelli la formula è vera nel mondo, $\forall\cM\in\cF, \forall w\in\cM,\cM,w\models A$. Quindi mi sembra che ci sia una sola direzione della dimostrazione, e ho commentato la direzione $\Leftarrow$}

\medskip

\noindent \emph{Case (\ref{en:AOBL})} -- Consider any frame $\cF$ that is $\mathbf{AFCP_{\OBL}}$-\emph{permitted} but such that $\cF \not\models \ax{AFCP_{\OBL}}$. This means that there exists a model $\cM = \tuple{W, \cNO , \cNP , V}$ based on $\cF$ such that $\cM \not\models \ax{AFCP_{\OBL}}$, i.e., there is a world $w\in W$ where
\begin{gather}
\cM , w \models \sPERM (p \vee q) \wedge \OBL \neg p \label{lemma:1AO}\\
\cM , w \not\models \sPERM q \label{lemma:2AO}
\end{gather}
By construction, from (\ref{lemma:2AO}) we have $\tset{q}\not\in \cNP (w)$, while from (\ref{lemma:1AO}) we have $\tset{p}\cup \tset{q}\in \cNP (w)$ and $W-\tset{p}\in \cNO (w)$, so $\cF$ is not $\mathbf{AFCP_{\OBL}}$-\emph{permitted}.

\medskip

\noindent \emph{Cases (\ref{en:APERM}) and (\ref{en:APERM2})} -- The proofs are similar to the one for Case (\ref{en:AOBL})	and are omitted.

\medskip

\noindent \emph{Case (\ref{en:IOBL})} --  As usual in these cases, the proof must show that, on the class of all $\mathbf{IFCP_{\OBL}}$-\emph{permitted} frames, for any model $\cM$ based on $\cF$  and for any world $w$ in it,
\[
\begin{array}{l}
\sPERM (p\vee q) \text{ and } \OBL\neg r \text{ are true in $\cM$ at }w \text{ and } r\to \neg q \text{ is valid in } \cF \Rightarrow  \\
\qquad \qquad \Rightarrow \sPERM p \text{ is true in $\cM$ at }w.
\end{array}
\]

\medskip

\noindent \emph{Cases (\ref{en:IPERM}) and (\ref{en:IPERM2})} -- The proofs are similar to the one for Case (\ref{en:IOBL}) and are omitted.
\end{proof}

The definitions of some basic notions and of canonical model for the classical bimodal logic $\Lg{E}$ (just consisting of $\ax{RE}$ for $\OBL$ and $\sPERM$) are standard.

In the rest of this section when we refer to a Deontic System $\Lg{S}$ we mean one the logic axiomatised in Section~\ref{sec:system}.

\begin{definition}[$\Lg{S}$-maximality]
A set $w$ is maximal iff it is $\Lg{S}$-consistent and for any formula $p$, either $p \in w$, or $\neg p \in w$.
\end{definition}

\begin{lemma}[Lindenbaum's Lemma]
For any Deontic System $\Lg{S}$, any consistent set $w$ of formulae can be extended to an $\Lg{S}$-maximal set $w^+$.
\end{lemma}

\begin{definition}[Canonical Model \cite{Chellas1980,Pacuit:2017}]\label{def:canonical}
\sloppy A canonical neighbourhood model $\cM=\langle W, \cNO, \cNP, V \rangle$ for any  system $\Lg{S}$ in our language $\cL$ (where $\Lg{S} \supseteq \Lg{E}$) is defined as follows:

\begin{enumerate}
\item $W$ is the set of all the $\Lg{S}$-maximal sets.

\item 
For any propositional letter $p$, $\|p\|_{\cM}:= | p |_{\Lg{S}}$,
where $| p |_{\Lg{S}} := \{ w \in W \mid p \in w \}$.

%
%

\item
If $\Box\in \set{\OBL , \sPERM}$, let $\cN_{\Box} := \bigcup_{w \in W} \cN_{\Box} (w)$ where for each world $w$, $\cN_{\Box} (w):= \{ \| a_i \|_{\cM} \mid \Box a_i \in w \}$.


\end{enumerate}
\end{definition}

\begin{lemma}[Truth Lemma \cite{Chellas1980,Pacuit:2017}]\label{truth_lemma}
If $\cM=\langle W, \cNO, \cNP, V \rangle$ is canonical for $\Lg{E}$, then
for any $w\in W$ and for any formula $p$, $p \in w$ iff $\cM , w\models p$.
\end{lemma}

Thus, we have as usual basic completeness result for $\Lg{E}$. To cover the other systems, it is enough to prove that all frame properties for the relevant schemata and rules are canonical.

\begin{lemma}\label{lemma:canonical_properties}
The frame properties of Definition \ref{def:frame-properties} are canonical.
\end{lemma}
\begin{proof}
The proofs for \textbf{$\Box$-\emph{supplementation}}, \textbf{$\wPERM$-\emph{coherence}}, and \textbf{$\sPERM$-\emph{coherence}} are standard.\\
\\*
\textbf{$\mathbf{AFCP_{\PERM}}$-\emph{permission}} -- Let us consider a canonical model $\cM$ for $\ax{AFCP_{\PERM}}$, any world $w$ in it, and any truth sets such that $\tset{p}\cup \tset{q}\in \cNP (w)$ and $W-\tset{q}\in \cNO (w)$. Clearly, $\tset{p\vee q}\in \cNP (w)$. Since $\ax{AFCP_{\PERM}}$ is valid (Lemma \ref{lemma:characterisation_results}), then $\sPERM p\in w$. By construction, this means that $\tset{p}\in \cNP (w)$, thus the model is $\mathbf{AFCP_{\PERM}}$-permitted.\\
\\*
\textbf{$\mathbf{AFCP_{\OBL}}$-\emph{permission}} and  \textbf{$\mathbf{AFCP2_{\PERM}}$-\emph{permission}}-- Similar to the case above.\\
\\*
\textbf{$\mathbf{IFCP_{\PERM}}$-\emph{permission}} -- Let us consider a canonical model $\cM$ for $\ax{IFCP_{\PERM}}$, any world $w$ in it, and any truth sets such that $\tset{p}\cup \tset{q}\in \cNP (w)$, $W-\tset{r}\not\in \cNO (w)$, and $W-\tset{s}\not\in \cNO (w)$. Clearly, $\tset{p\vee q}\in \cNP (w)$. Also, assume $\tset{r} \subseteq \tset{p}$ and $\tset{r}\subseteq \tset{q}$. Since $\ax{IFCP_{\PERM}}$ is valid (Lemma \ref{lemma:characterisation_results}) then $\sPERM p\wedge \sPERM q\in w$. By construction, this means that $\tset{p}\in \cNP (w)$ and $\tset{q}\in \cNP (w)$, thus the model is $\mathbf{IFCP_{\PERM}}$-permitted.\\
\\*
\textbf{$\mathbf{IFCP_{\OBL}}$-\emph{permission}} and  \textbf{$\mathbf{IFCP2_{\PERM}}$-\emph{permission}}-- Similar to the case above.
\end{proof}

Hence, the following result is ensured.

\completeness*

Finally, let us prove Corollary \ref{cor:inclusion}.

\inclusion*
\begin{proof}
\emph{Case \eqref{cor:1}} --
For $\Lg{FCP_2}\subseteq\Lg{FCP_1}$ we first notice that for every formula $p$, $\vdash p\to p$; hence,
axioms $\ax{AFCP_\OBL}$ and $\ax{AFCP_\PERM}$ can be considered as simple instances of $\ax{IFCP_\OBL}$ and $\ax{ICFP_\PERM}$ respectively.

To show that the inclusion is strict the model below provides an $\ax{AFCP_\OBL}$-permitted model
that does not validate $\ax{IFCP_\OBL}$.

Let $\cM=\langle W,\cNO,\cNP,V\rangle$, where:
\begin{itemize}
  \item $W=\set{w_1, w_2, w_3, w_4,w_5}$;
  \item $V(a)=\set{w_1,w_4,w_5}$,
        $V(b)=\set{w_2,w_3,w_4}$ and
        $V(c)=\set{w_1,w_2}$;
  \item $\cN_{\OBL}(w_1)=\set{\set{w_4}}$; and
  \item $\cN_{\PERM}(w_1)=\set{\set{w_1,w_2,w_3}}$.
\end{itemize}
It is easy to verify that the model is $\ax{AFCP_\OBL}$-permitted, $\sPERM(\neg a\vee c)$ and $\OBL(a\wedge c)$ are true in $w_1$: $\tset{\neg a\vee c}=\set{w_1,w_2,w_3}\in\cNP(w_1)$ and $\tset{a\wedge c}=\set{w_4}\in\cNO(w_1)$, and clearly $a\to \neg\neg a$. However, $\tset{c}=\set{w_1,w_2}\notin\cNP(w_1)$.

\medskip

For $\Lg{FCP_1}\subseteq\Lg{FCP_3}$ it is enough to prove the following showing that $\ax{IFCP_\OBL}$
and $\ax{IFCP_\PERM}$ are derivable from $\ax{AFCP_\OBL}$ and $\ax{AFCP_\PERM}$ and $\ax{RM}$. ($\ax{RM}$ is
valid in every classical modal logic containing $\ax{M}$ \cite{Chellas1980}.)
\[
\begin{array}{lll}
  1. & \sPERM(p\vee q)\wedge \OBL r & Hyp.\\
  2. & r \to \neg p                 & Hyp.\\
  3. & \sPERM(p\vee q)              & 1,\ \CPL\\
  4. & \OBL r                       & 1,\ \CPL\\
  5. & \OBL r \to \OBL\neg p        & 2,\ \ax{RM}\\
  6. & \OBL\neg p                   & 4,5,\ \CPL\\
  7. & (\sPERM(p\vee q)\wedge \OBL\neg p) \to \sPERM q & \ax{AFCP_\OBL}\\
  8. & \sPERM q                     & 3,6,7,\ \CPL
\end{array}
\]

\[
\begin{array}{rll}
  1. & \sPERM(p\vee q)\wedge \wPERM r \wedge \wPERM s & Hyp.\\
  2. & r \to p        & Hyp.\\
  3. & s \to q        & Hyp.\\
  4. & \wPERM r \to \wPERM p    & 2,\ \ax{RM}\\
  5. & \wPERM s \to \wPERM q    & 3,\ \ax{RM}\\
  6. & \wPERM r        & 1,\ \CPL\\
  7. & \wPERM s        & 1,\ \CPL\\
  8. & \sPERM(p\vee q) & 1,\ \CPL\\
  9. & \wPERM p        & 4,6,\ \CPL\\
 10. & \wPERM q        & 5,7,\ \CPL\\
 11. & (\sPERM(p\vee q)\wedge \wPERM r \wedge \wPERM s) \to (\sPERM p \wedge \sPERM q) & \ax{AFCL_\PERM}\\
 12. & \sPERM p \wedge \sPERM q  & 8,9,10,\ \CPL
\end{array}
\]

If we take the model used in the previous case and we add $\set{w_1,w_2}$ to $\cNP(w_1)$, obtaining $\cNP(w_1)=\set{\set{w_1,w_2,w_3}, \set{w_1,w_2}}$, then we have a non $\OBL$-supplemented $\ax{IFCP_\OBL}$-permitted model falsifying $\ax{M}$.  Now the model is $\ax{IFCP_\OBL}$-permitted, but not $\OBL$-supplemented: $\tset{a\wedge c}=\set{w_4}\in\cNO(w_1)$, $\set{w_4}=\tset{a}\cap\tset{c}$, but $\tset{a},\tset{c}\notin\cNO(w_1)$, falsifying the following instance of $\ax{M}$: $\OBL(a\wedge c)\to\OBL a\wedge \OBL c$.

\medskip

For $\Lg{FCP_2}\subset\Lg{FCP_4}$, and $\Lg{FCP_1}\subset\Lg{FCP_5}$ it is immediate to verify that $\ax{AFCP2}_{\PERM}$ implies $\ax{AFCP}_{\PERM}$ in $\Lg{CPL}$ but not the other way around; similarly for $\ax{IFCP2}_{\PERM}$ and $\ax{IFCP}_{\PERM}$.

\medskip

The second derivation given in the case of $\Lg{FCP_1}\subset\Lg{FCP_3}$  can be trivially adjusted to show that $\Lg{FCP_5}\subseteq\Lg{FCP_6}$, for the strictness of the inclusion we can reuse the model given in the same case.

\medskip

For $\Lg{FCP_4}\subseteq\Lg{FCP_5}$, we can reuse the derivation that shows that $\ax{IFCP}_{\PERM}$ is a derived rule in $\Lg{FCP_3}$, using $\sPERM(p\vee q)\wedge\wPERM r$ in step 1. and axiom $\ax{AFCP2}_{\PERM}$ in step 11.

To show that the inclusion is strict consider the model $\cM=\langle W, \cNO, \cNP, V\rangle$, where
\begin{itemize}
  \item $W=\set{w_1,w_2,w_3,w_4}$,
  \item $\cNO(w_1)=\set{\set{w_1}}$,
  \item $\cNP(w_1)=\set{\set{w_1,w_2,w_3},\set{w_1,w_2}}$,
  \item $V(a)=\set{w_1,w_2}$, $V(b)=\set{w_1,w_4}$, $V(c)=\set{w_2,w_3}$.
\end{itemize}
It is immediate to verify that the model is $\ax{IFCP2_\PERM}$-permitted but not $\ax{AFCP2_\PERM}$-permitted; indeed, $\cM,w_1\models \sPERM(a \vee c) \wedge \wPERM(a\wedge b)$ but $\cM,w_1\not\models \sPERM a$,

\medskip

For $\Lg{FCP_3}\subset\Lg{FCP_6}$ the result follows immediately from the fact that $\Lg{FCP_2}\subset\Lg{FCP_4}$.

\medskip

\noindent \emph{Case (\ref{cor:2})} -- The result follows from Case \eqref{cor:1} above and Theorem \ref{th:completeness}.
\end{proof}

\end{document}

%% file: macros.tex
\usepackage{amsmath,amsthm}
\usepackage{thmtools,thm-restate}
\usepackage{graphicx,xcolor}
\usepackage{latexsym}
\usepackage{amssymb,setspace}
\usepackage[numbers,sort]{natbib}
\usepackage{paralist}
\usepackage[small,compact]{titlesec}
\usepackage{xparse}
\usepackage{microtype}

   \newtheoremstyle{example}{\topsep}{\topsep}%
     {}
     {}
     {\bfseries}
     {}
     {\newline}
     {\thmname{#1}\thmnumber{ #2}\thmnote{ #3}}

   \theoremstyle{example}
   \newtheorem{example}{Example}[section]
   \theoremstyle{definition}
   \newtheorem{definition}{Definition}[section]
   \newtheorem{lemma}{Lemma}[section]

\newcommand{\ax}[1]{\mathbf{#1}}

\newcommand{\Lg}[1]{\mathbf{#1}}
\newcommand{\SDL}{\mathbf{SDL}}
\newcommand{\CPL}{\mathbf{CPL}}

\newcommand{\equaldef}{=_{\mathit{def}}}

\newcommand{\PERM}{\mathbf{P}}
\newcommand{\sPERM}{\mathbf{P_{\kern-1.5pt s}}}
\newcommand{\wPERM}{\mathbf{P_{\kern-2pt w}}}
\newcommand{\OBL}{\mathbf{O}}

\newcommand{\prop}{\mathrm{PROP}}

\newcommand{\cM}{\mathcal{M}}
\newcommand{\cN}{\mathcal{N}}
\newcommand{\cNP}{\mathcal{N}_{\PERM}}
\newcommand{\cNO}{\mathcal{N}_{\OBL}}
\newcommand{\cF}{\mathcal{F}}
\newcommand{\cC}{\mathcal{C}}
\newcommand{\cL}{\mathcal{L}}

\newcommand{\set}[1]{\left\{#1\right\}}
\newcommand{\tuple}[1]{\langle #1\rangle}

\NewDocumentCommand{ \tset }{ m O{M} }
  { ||#1||\sb{\mathcal{#2}} }

%% file: GovRotFCP.bbl
\begin{thebibliography}{23}
\providecommand{\natexlab}[1]{#1}
\providecommand{\url}[1]{\texttt{#1}}
\expandafter\ifx\csname urlstyle\endcsname\relax
  \providecommand{\doi}[1]{doi: #1}\else
  \providecommand{\doi}{doi: \begingroup \urlstyle{rm}\Url}\fi

\bibitem[Alchourr\'{o}n and Bulygin(1971)]{alchourron71normative}
C.~E. Alchourr\'{o}n and E.~Bulygin.
\newblock \emph{Normative Systems}.
\newblock Springer Verlag, 1971.

\bibitem[Alchourr\'on and Bulygin(1981)]{AB:1981}
Carlos~E. Alchourr\'on and Eugenio Bulygin.
\newblock The expressive conception of norms.
\newblock In Risto Hilpinen, editor, \emph{New Studies in Deontic Logic}, pages
  95--125. D. Reidel, Dordrecht, 1981.

\bibitem[Alchourr\'on and Bulygin(1984)]{alchourron-bulygin:1984}
Carlos~E. Alchourr\'on and Eugenio Bulygin.
\newblock Permission and permissive norms.
\newblock In W.~Krawietz et~al., editor, \emph{Theorie der Normen}. Duncker \&
  Humblot, 1984.

\bibitem[Anglberger et~al.(2014)Anglberger, Dong, and Roy]{RoyDeon}
Albert J.~J. Anglberger, Huimin Dong, and Olivier Roy.
\newblock Open reading without free choice.
\newblock In Fabrizio Cariani, Davide Grossi, Joke Meheus, and Xavier Parent,
  editors, \emph{Deontic Logic and Normative Systems}, pages 19--32, Cham,
  2014. Springer International Publishing.

\bibitem[Anglberger et~al.(2015)Anglberger, Gratzl, and
  Roy]{anglberger_gratzl_roy_2015}
Albert~J.J. Anglberger, Nobert Gratzl, and Olivier Roy.
\newblock Obligation, free choice, and the logic of weakest permissions.
\newblock \emph{The Review of Symbolic Logic}, 8\penalty0 (4):\penalty0
  807–827, 2015.

\bibitem[Asher and Bonevac(2005)]{Asher2005}
Nicholas Asher and Daniel Bonevac.
\newblock Free choice permission is strong permission.
\newblock \emph{Synthese}, 145\penalty0 (3):\penalty0 303--323, 2005.

\bibitem[Barker(2010)]{barker:2010:permission}
Chris Barker.
\newblock Free choice permission as resource-sensitive reasoning.
\newblock \emph{Semantics and Pragmatics}, 3\penalty0 (10):\penalty0 1--38,
  2010.

\bibitem[Blackburn et~al.(2001)Blackburn, de~Rijke, and Venema]{Blackburn2001}
Patrick Blackburn, Maarten de~Rijke, and Yde Venema.
\newblock \emph{Modal Logic}.
\newblock Cambridge University Press, 2001.

\bibitem[Broersen(2004)]{Broersen:2004}
Jan Broersen.
\newblock Action negation and alternative reductions for dynamic deontic
  logics.
\newblock \emph{Journal of Applied Logic}, 2\penalty0 (1):\penalty0 153 -- 168,
  2004.

\bibitem[Chellas(1980)]{Chellas1980}
Brian~F. Chellas.
\newblock \emph{Modal logic: an introduction}.
\newblock Cambridge University Press, 1980.

\bibitem[Dong(2017)]{Dong:2017}
Huimin Dong.
\newblock \emph{Permission in Non-Monotonic Normative Reasoning}.
\newblock PhD thesis, University of Bayreuth, 2017.

\bibitem[Governatori and Rotolo(2006)]{ajl:ctd}
Guido Governatori and Antonino Rotolo.
\newblock Logic of violations: A {Gentzen} system for reasoning with
  contrary-to-duty obligations.
\newblock \emph{Australasian Journal of Logic}, 4:\penalty0 193--215, 2006.

\bibitem[Governatori et~al.(2016)Governatori, Olivieri, Calardo, and
  Rotolo]{deon:2016}
Guido Governatori, Francesco Olivieri, Erica Calardo, and Antonino Rotolo.
\newblock Sequence semantics for norms and obligations.
\newblock In \emph{Proceedings of {DEON} 2016}, London, 2016. College
  Publications.

\bibitem[Hansson(2013)]{Hansson:PERM}
Sven~O. Hansson.
\newblock The varieties of permissions.
\newblock In D.~Gabbay, J.~Horty, X.~Parent, R.~van~der Meyden, and L.~van~der
  Torre, editors, \emph{Handbook of Deontic Logic and Normative Systems}.
  College Publications, 2013.

\bibitem[Kamp(1973)]{Kamp:1973}
Hans Kamp.
\newblock Free choice permission.
\newblock \emph{Proceedings of the Aristotelian Society}, 74\penalty0
  (n/a):\penalty0 57--74, 1973.

\bibitem[Lewis(1979)]{Lewis:1979}
David Lewis.
\newblock \emph{A Problem About Permission}, pages 163--175.
\newblock Springer Netherlands, Dordrecht, 1979.

\bibitem[Lokhorst(1997)]{Lokhorst}
Gert-Jan~C. Lokhorst.
\newblock {Deontic Linear Logic with Petri Net Semantics}.
\newblock Technical report, FICT, Center for the Philosophy of Information and
  Communication Technology. Rotterdam, 1997.

\bibitem[Makinson(1999)]{Makinson:1999}
D.~Makinson.
\newblock On a fundamental problem of deontic logic.
\newblock In Paul McNamara and Henry Prakken, editors, \emph{Norms, Logics and
  Information Systems. New Studies in Deontic Logic and Computer Science},
  pages 29--54. IOS Press, Amsterdam, 1999.

\bibitem[Makinson and van~der Torre(2003)]{makinson-torre:2003}
David Makinson and Leendert van~der Torre.
\newblock Permission from an input/output perspective.
\newblock \emph{Journal of Philosophical Logic}, 32\penalty0 (4):\penalty0
  391--416, 2003.

\bibitem[McNamara(2018)]{sep-logic-deontic}
Paul McNamara.
\newblock Deontic logic.
\newblock In Edward~N. Zalta, editor, \emph{The Stanford Encyclopedia of
  Philosophy}. Metaphysics Research Lab, Stanford University, fall 2018
  edition, 2018.

\bibitem[Pacuit(2017)]{Pacuit:2017}
Eric Pacuit.
\newblock \emph{Neighborhood Semantics for Modal Logic}.
\newblock Springer, 2017.

\bibitem[von Wright(1963)]{vonwright:1963}
Georg~Henrik von Wright.
\newblock \emph{Norm and action: {A} logical inquiry}.
\newblock Routledge and Kegan Paul, 1963.

\bibitem[von Wright(1968)]{vonWright1968-VONAEI-3}
Georg~Henrik von Wright.
\newblock \emph{An Essay in Deontic Logic and the General Theory of Action with
  a Bibliography of Deontic and Imperative Logic}.
\newblock North-Holland Pub. Co, 1968.

\end{thebibliography}
